\documentclass[prodmode,acmteac]{acmsmall}

\usepackage[ruled]{algorithm2e}
\usepackage[numbers]{natbib}
\usepackage{amssymb,appendix,pstricks,verbatim,enumerate,pst-node,amsfonts,amsmath,algorithm2e}
\def\eps{\epsilon}

\SetAlFnt{\small}
\SetAlCapFnt{\small}
\SetAlCapNameFnt{\small}
\SetAlCapHSkip{0pt}
\IncMargin{-\parindent}

\begin{document}

\markboth{R. Cole, V. Gkatzelis, and G. Goel}{Mechanism Design for Fair Division}

\title{Mechanism Design for Fair Division:\\ \indent {\small Allocating Divisible Items without Payments}}
\author{RICHARD COLE
\affil{Courant Institute, New York University}
VASILIS GKATZELIS
\affil{Courant Institute, New York University}
GAGAN GOEL
\affil{Google Research, New York}}

\begin{abstract}
We revisit the classic problem of fair division from a mechanism design perspective, using {\em Proportional Fairness} 
as a benchmark. In particular, we aim to allocate a collection of divisible items to a set of agents while incentivizing
the agents to be truthful in reporting their valuations.
For the very large class of homogeneous valuations, we design a truthful mechanism that provides {\em every agent} with 
at least a $1/e\approx 0.368$ fraction of her Proportionally Fair valuation. To complement this result, we show that no
truthful mechanism can guarantee more than a $0.5$ fraction, even for the restricted class of additive linear valuations.
We also propose another mechanism for additive linear valuations that works really well when every item is highly demanded.
To guarantee truthfulness, our mechanisms discard a carefully chosen fraction of the allocated resources;
we conclude by uncovering interesting connections between our mechanisms and known mechanisms that use money instead.
\end{abstract}





\begin{bottomstuff}
This work was supported in part by NSF grants CCF-0830516
and  CCF-1217989. The second author was an intern at Google when part of this work took place.
\end{bottomstuff}

\maketitle

\section{Introduction}

From inheritance and land dispute resolution to treaty negotiations and divorce settlements,
the problem of fair division of diverse resources has troubled man since antiquity. Not
surprisingly, it has now also found its way into the highly automated, large scale world
of computing. As the leading internet companies 
guide
the paradigm shift into
cloud computing, more and more services that used to be run on isolated machines are being
migrated to shared computing clusters. Moreover, instead of just human beings bargaining or negotiating,
one now also finds programmed strategic agents 
seeking
resources. The goal of the resulting multiagent resource allocation problems~\cite{MARA} is to find solutions
that are fair to the agents without introducing unnecessary inefficiencies.

One of the most challenging facets of this change is the need for higher quality incentive
design in the form of protocols or mechanisms. As the peer-to-peer revolution has taught us,
a proper set of incentives can make or break a system as the number of agents grows~\cite[Chapter 23]{AGT}.
We therefore revisit this classic fair division problem from a purely mechanism design approach, aiming to
create simple and efficient mechanisms that are not susceptible to strategic manipulation 
by
the participating agents; in particular, we want to design {\em truthful} mechanisms for fair division of heterogeneous goods.

One distinguishing property of resource allocation protocols in computing
is that, more often than not, they need to eschew monetary transfers completely.
This is so because, for instance, agents could represent internal teams in an
internet company 
which
are competing for resources. This, of course,
severely limits what the mechanism designer can achieve since the collection of payments
is the most 
versatile
method for designing truthful mechanisms.
In light of this, essentially the only tool left for aligning the agents'
incentives with the objectives of the system is what Hartline and Roughgarden referred to as ``money
burning''~\cite{HR08}. That is, the system can choose to intentionally degrade the quality of its services
(in our case this will mean discarding resources) in order to influence the preferences of the agents.
This degradation of service can often be interpreted as an implicit form of ``payment'', but since
these payments do not correspond to actual trades, they are essentially burned or used for other purposes.

But even before dealing with the fact that the participating agents may behave strategically, one
first needs to ask what 
is
the right objective for fairness.
This question alone has been the subject of long debates,
in both social science and game theory, leading to a very rich literature. At the time of writing this paper, there are
{\em five} academic books ~\cite{Young94,BramsT96,RobertsonWebb98,Moulin03,Barbanel04} written on the topic
of fair division, providing an overview of various proposed solutions for fairness. In this
paper we will be focusing on resources that are {\em divisible}; for such settings, the most attractive
solution for efficient and fair allocation is the Proportionally Fair solution (PF).
In brief, a PF allocation is a Pareto optimal allocation $x^*$ which compares
favorably to any other Pareto optimal allocation $x$ in the sense that, when switching from $x$
to $x^*$, the aggregate percentage gain in happiness of the agents outweighs the aggregate percentage loss.
The notion of PF was first introduced in the seminal work of \citet{Kelly}  in the context
of TCP congestion control. Since then it has become the de facto solution for bandwidth sharing
in the networking community, and is in fact the {\em most widely} implemented solution in practice
(for instance see~\cite{AndrewsQS05})\footnote{We note that some of the earlier work
on Proportional Fairness such as \cite{Kelly} and \cite{K+} have 2000+ and 3900+ citations respectively in
Google Scholar, 
indicating
the importance and usage of this solution.}. The wide adoption of PF
as the solution for fairness is not a fluke, but is grounded in the fact that PF is equivalent to the
Nash bargaining solution~\cite{Nash50}, and to the Competitive Equilibria with Equal Incomes
(CEEI)~\cite{HZ79, Varian74, Eisenberg61} for a large class of valuation functions. Both Nash bargaining
and the CEEI are well regarded solutions in microeconomics for bargaining and fairness.

A notable property of the PF solution is that it gives a
good tradeoff between fairness and efficiency. One extreme notion of fairness is the Rawlsian notion of
the egalitarian social welfare that aims to maximize the quality of service of the least satisfied
agent irrespective of how much inefficiency this might be causing. On the other extreme, the utilitarian
social welfare approach aims to maximize efficiency while disregarding how unsatisfied some agents might become.
The PF allocation lies between these two extremes by providing a significant fairness guarantee without
neglecting efficiency. As we showed in a recent work~\cite{CGG13}, for instances with just two players
who have affine valuation functions the PF allocation has a social welfare of at least $0.933$ times the optimal one.

Unfortunately, the PF allocation has one significant drawback: it cannot 
be implemented using truthful mechanisms
without the use of payments; even for simple instances involving just two agents
and two items, it is not difficult to show that no truthful mechanism can obtain a PF
solution. This motivates the following natural question: can one design truthful mechanisms that yield a good approximation to the PF solution?
Since our goal is 
to obtain
a fair division, we seek a strong notion of approximation in which every agent gets a good approximation of her PF valuation.
One of our main results is to give a truthful mechanism which guarantees that every agent will receive at least a $1/e$ fraction of her PF valuation
for a very large class of valuation functions. We note that this is one of the very few positive
results in multi-dimensional mechanism design without payments. We 
demonstrate the
hardness of achieving
such truthful approximations by providing 
an almost matching negative result for a 
restricted class of valuations.

While a $1/e$ approximation factor is quite surprising for such a general setting,
in some circumstances one would prefer to restrict the setting in order to achieve
a ratio much closer to 1.
Our final result concerns such a scenario, which is motivated by the real-world
privatization auctions that took place in Czechoslovakia in the early 90s.
At that time, 
%
the Czech government sought to privatize the state
owned firms dating
from the then recently ended communist era. The government's goal was
two-fold --- first, to distribute shares of these companies to their
citizens in a fair manner, and second,
to calculate the market prices of these
companies so that the shares could be traded in the open market after
the initial allocation. To this end, they ran an auction, as described
in~\cite{AH2000}.
Citizens could choose to participate by buying 1000 vouchers at a cost
of 1,000 Czech Crowns, about \$35, a fifth of the average monthly salary.
Over 90\% of those eligible participated.
These vouchers were then used to bid for shares in the available 1,491 firms.
We believe that the PF allocation provides a very appropriate solution
for this example, both to calculate a fair allocation and to compute
market prices.
Our second mechanism solves the problem of finding 
allocations very close to the PF allocation
in a truthful fashion for such natural scenarios where there is high demand for each resource.

\subsection{Our results}
In this work we provide some surprising positive results for the problem of multi-dimensional mechanism
design without payments. We focus on allocating divisible items and we use the widely
accepted solution of proportional fairness as the benchmark regarding the valuation
that each participating player deserves. In this setting, we undertake the design of truthful mechanisms
that approximate this solution; we consider a strong notion of approximation, requiring that every
player receives a good fraction of the valuation that she deserves according to the proportionally
fair solution of the instance at hand.

The main contribution of this paper is the {\em Partial Allocation} mechanism.
In Section~\ref{sec:PAmech} we analyze this mechanism and we prove that it is truthful and 
that it guarantees that every player will receive at least a $1/e$ fraction of her proportionally 
fair valuation. These results hold for the very general class of instances with players
having arbitrary homogeneous valuation functions of degree one. This includes
a wide range of well studied valuation functions, from additive linear and Leontief,
to Constant Elasticity of Substitution and Cobb-Douglas~\cite{MWG}.
We later show that for the cases of additive linear and Leontief valuation functions 
the outcomes of this mechanism satisfy envy-freeness. Also, we extend both
the approximation and the truthfulness guarantees to instances with homogeneous 
valuations of any degree. To complement these positive results, we provide a negative 
result showing that no truthful mechanism can guarantee to every player an allocation with 
value greater than $0.5$ of the value of the PF allocation, even if the mechanism
is restricted to the class of additive linear valuations.
	
In Section~\ref{sec:PFapprox} we show that restricting the set of possible instances to ones involving
players with additive linear valuations\footnote{Note that our negative results imply that the restriction
to additive linear valuations alone would not be enough to enable significantly better approximation factors.}
and items with high prices in the competitive equilibrium from
equal incomes\footnote{The prices induced by the market equilibrium  when all bidders have a
unit of scrip money; also referred to as PF prices.}
will actually enable the design of even more efficient and useful mechanisms.
%
%
We present the {\em Strong Demand Matching} (SDM) mechanism, a truthful mechanism that performs
increasingly well as the competitive equilibrium prices increase.
More specifically, if $p_j^*$ is the price of item $j$, then the approximation factor guaranteed
by this mechanism is equal to $\min_j \left(p_j^*/\left\lceil p_j^* \right\rceil\right)$.
It is interesting to note that scenarios such as the privatization auction
mentioned above involve a number of bidders much larger than the number of items;
as a rule, we expect this to lead to high prices and a very good approximation of the participants' PF valuations.

Finally, in Section~\ref{sec:withmoney} we provide interesting connections between the two
mechanisms that we propose and well known mechanisms that use monetary payments. Specifically,
we reveal a connection between the amount of resources being discarded and monetary payments.
In a nutshell, multiplicative reductions in the bidders' final allocations turn out to have an 
effect which is analogous to monetary payments. As a result, we anticipate that this approach may have 
a significant impact on other problems in mechanism design without money. Indeed, we have already 
applied this approach to the problem
of maximizing social welfare without payments for which a special two-agent version of the
Partial Allocation mechanism allowed us to improve upon a setting for which mostly negative
results were known~\cite{CGG13}.

\subsection{Related Work}
Our setting is closely related to the large topic of fair division or
cake-cutting~\cite{Young94,BramsT96,RobertsonWebb98,Moulin03,Barbanel04},
which has been studied since the 1940's, using the $[0,1]$ interval as the standard representation
of a cake. Each agent's preferences take the form of a valuation function over this interval, and
then the valuations of unions of subintervals 
are additive.
Note that the class of homogeneous valuation functions of degree one takes us beyond this standard cake-cutting
model. Leontief valuations for example, allow for complementarities in the valuations,
and then the valuations of unions of subintervals need not be additive.
On the other hand, the additive linear valuations setting that we focus on in Section~\ref{sec:PFapprox}
is very closely related to cake-cutting with piecewise constant valuation functions over the $[0,1]$ interval.
Other common notions of fairness that have been studied in this literature are,
proportionality\footnote{It is worth distinguishing the notion of PF from that of proportionality
by noting that the latter is a much weaker notion, directly implied by the former.}, envy-freeness,
and equitability~\cite{Young94,BramsT96,RobertsonWebb98,Moulin03,Barbanel04}.

Despite the extensive work on fair resource allocation, truthfulness considerations have not played a major
role in this literature. Most results related to truthfulness were weakened by the assumption that each agent would be
truthful in reporting her valuations unless this strategy was dominated.
Very recent work~\cite{CLPP10,MosselT10,ZivanDOS10,MN12} studies truthful cake cutting variations using the
standard notion of truthfulness according to which an agent 
need
not be truthful unless doing so is a dominant strategy.
\citet{CLPP10} study truthful cake-cutting with agents having piecewise uniform valuations and
they provide a polynomial-time mechanism that is truthful, proportional, and envy-free. They also design randomized
mechanisms for more general families of valuation functions, while \citet{MosselT10} prove the existence
of truthful (in expectation) mechanisms satisfying proportionality in expectation for general valuations.
\citet{ZivanDOS10} aim to achieve envy-free Pareto optimal
allocations of multiple divisible goods while reducing, but not eliminating, the agents' incentives to lie.
The extent to which untruthfulness is reduced by their proposed mechanism is only evaluated empirically and
depends critically on their assumption that the resource limitations are soft constraints.
Very recent work by \citet{MN12} provides evidence that truthfulness comes at a significant
cost in terms of efficiency.

The recent papers of \citet{GC10} and of \citet{HanSTZ11} also consider the
truthful allocation of multiple divisible goods; they focus on additive linear valuations and their
goal is to maximize the social welfare (or efficiency) after scaling every player's reported valuations
so that her total valuation for all items is 1.
\citet{GC10} study two-agent instances, providing both upper and lower bounds
for the achievable approximation; \citet{HanSTZ11} extend these results and also study the multiple agents setting.
For problem instances that may involve an arbitrary number of items both papers provide negative results:
no non-trivial approximation factor can be achieved by any truthful mechanism when the number of players is also unbounded.
For the two-player case, after \citet{GC10} studied some classes of dictatorial
mechanisms, \citet{HanSTZ11} showed that no dictatorial mechanism can guarantee more
than the trivial $0.5$ factor. Interestingly, we recently showed~\cite{CGG13} that combining
a special two-player version of the Partial Allocation mechanism with a dictatorial mechanism
can actually beat this bound, achieving a $2/3$ approximation.

The resource allocation literature has seen a resurgence of work studying fair and efficient
allocation for Leontief valuations~\cite{GZH11,DFH12,PPS12,GN12}. These valuations exhibit perfect
complements and they are considered to be natural valuation abstractions for computing settings where
jobs need resources in fixed ratios. \citet{GZH11} defined the notion
of Dominant Resource Fairness (DRF), which is a generalization of the egalitarian social welfare
to multiple types of resources. This solution has the advantage that it can be implemented
truthfully for this specific class of valuations; as the authors acknowledge, the CEEI solution
would be the preferred fair division mechanism in that setting as well, and 
its main drawback is
the fact that it cannot be implemented truthfully.
\citet{PPS12} assessed DRF in terms of the resulting efficiency, showing that it performs
poorly. \citet{DFH12} proposed an alternate fairness criterion called
Bottleneck Based Fairness, which \citet{GN12} subsequently
showed is satisfied by the proportionally fair allocation. \citet{GN12}
also posed the study of incentives related to this latter notion as an interesting open problem.
Our results could potentially have significant impact 
on this line of work as we are providing
a truthful way to approximate a solution which is recognized as a good benchmark.
It would also be interesting to study the extent to which the Partial Allocation mechanism
can outperform the existing ones in terms of efficiency.

Most of the papers mentioned above contribute to our understanding of the trade-offs between either truthfulness
and fairness, or truthfulness and social welfare. Another direction that has been actively pursued is to
understand and quantify the interplay between fairness and social welfare. \citet{CKKK12}
measured the deterioration of the social welfare caused by different fairness restrictions, the price of
fairness. More recently, \citet{CLPP11} designed algorithms for computing allocations that (approximately)
maximize social welfare while satisfying envy-freeness.  Also, the work of \citet{CFFKO11}, and of \citet{FL12}
considers the question of finding mechanisms that satisfy both truthfulness and envy-freeness without sacrificing
any social welfare.

Our results fit into the general agenda of approximate mechanism design without money, explicitly initiated by
\citet{ProcacciaT09}. More interestingly, the underlying connection with VCG payments proposes
a framework for designing truthful mechanisms without money and we anticipate that this might have a significant
impact on this literature.



\section{Preliminaries}\label{sec:prelim}
Let $M$ denote the set of $m$ items and $N$ the set of $n$ bidders. Each
item is divisible, meaning that it can be divided into arbitrarily small
pieces,
which are
then allocated to different bidders. An allocation $x$ of these
items to the bidders defines the fraction $x_{ij}$ of each item $j$ that each bidder
$i$ will be receiving; let $\mathcal{F}=\left\{x~|~x_{ij}\geq 0 \text{ and } \sum_i x_{ij}\leq 1 \right\}$
denote the set of feasible allocations. Each bidder is assigned a weight $b_i\geq 1$ which allows for interpersonal
comparison of valuations, and can serve as priority in computing applications, as clout in
bargaining applications, or as a budget for the market equilibrium interpretation of our results.
We assume that $b_i$ is defined by the mechanism as it cannot be truthfully elicited 
from the bidders.
The preferences of each bidder $i\in N$ take the form of a valuation function
$v_i(\cdot)$, that assigns nonnegative values to every allocation in $\mathcal{F}$.
We assume that every player's valuation for a given allocation $x$ only depends on
the bundle of items that she will be receiving.

We will present our results assuming that the valuation functions
are homogeneous of degree one, i.e.\ player $i$'s valuation for an allocation $x' = f\cdot x$
satisfies
$v_i(x')=f\cdot v_i(x)$, for any scalar $f>0$. We later discuss how to extend these results to
general homogeneous valuations of degree $d$ for which $v_i(x')=f^d\cdot v_i(x)$.
A couple of interesting examples of homogeneous valuations functions of degree one are additive linear valuations
and Leontief valuations;
according to the former, every player has a valuation $v_{ij}$ for each item $j$ and $v_i(x) = \sum_j x_{ij} v_{ij}$,
and according to the latter, each player $i$'s type corresponds to a set of values $a_{ij}$, one for each item, and
$v_i(x) = \min_j\left\{x_{ij}/a_{ij}\right\}$. (i.e.\ player $i$ desires the items in the ratio $a_{i1}:a_{i2}: \ldots :a_{im}$.)

An allocation $x^*\in \mathcal{F}$ is \emph{Proportionally Fair} (PF) if, for any other allocation $x'\in\mathcal{F}$
the (weighted) aggregate proportional change to the valuations after replacing $x^*$ with $x'$ is not positive, i.e.:
\begin{equation}\label{ineq:PF}
\sum_{i\in N}{\frac{b_i[v_i(x')-v_i(x^*)]}{v_i(x^*)}}\leq 0.
\end{equation}
This allocation rule is a strong refinement of Pareto efficiency, since Pareto efficiency only guarantees that if some player's proportional
change is strictly positive, then there must be some player whose proportional change is negative.
The Proportionally Fair solution can also be defined as an allocation $x\in\mathcal{F}$ that maximizes $\prod_i{[v_i(x)]^{b_i}}$, or
equivalently (after taking a logarithm), that maximizes $\sum_i b_i \log v_i(x)$; we will refer to these two equivalent
objectives as the PF objectives. Note that, although the PF allocation need not be unique for a given instance,
it does provide unique bidder valuations~\cite{EisenbergG59}.

We also note that the PF solution is equivalent to the Nash bargaining solution. John Nash in his seminal paper ~\cite{Nash50} considered an axiomatic
approach to bargaining and gave four axioms that any bargaining solution must satisfy. He showed that these four axioms 
yield
a unique solution
which is captured by a convex program; this convex program is equivalent to the one defined above for the PF solution. Another well-studied allocation
rule which is equivalent to the PF allocation is the {\em Competitive Equilibrium}. \citet{Eisenberg61} showed that if all agents have valuation
functions that are quasi-concave and homogeneous of degree 1, then the competitive equilibrium is also captured by the same convex program as the one for the PF
solution. The {\em Competitive Equilibrium with Equal Incomes} (CEEI) has been proposed as the ideal allocation rule for fairness in microeconomics ~\cite{Varian74,HZ79, B10, OSB10}.


Given a valuation function reported from each bidder, we want to design mechanisms that output
an allocation of items to bidders. We restrict ourselves to truthful mechanisms, i.e.\ mechanisms
such that any false report from a bidder will never return her a more valuable allocation.
Since proportional fairness cannot be implemented via truthful mechanisms, we will measure the performance
of our mechanisms based on the extent to which they approximate this benchmark.
More specifically, the approximation factor, or competitive factor
of a mechanism will correspond to the minimum value
of $\rho(\mathcal{I})$ across all relevant instances $\mathcal{I}$, where
\[ \rho(\mathcal{I}) = \min_{i\in N}\left\{ \frac{v_i(x)}{v_i(x^*)} \right\},\]
and $x,x^*$ are the allocation generated by the mechanism for instance $\mathcal{I}$
and the PF allocation of $\mathcal{I}$ respectively.

\section{The Partial Allocation Mechanism}\label{sec:PAmech}

In this section, we define the {\em Partial Allocation} (PA) mechanism as a novel way to allocate
divisible items to bidders with homogeneous valuation functions of degree one. We subsequently prove
that this non-dictatorial mechanism not only achieves truthfulness, but also guarantees that
every bidder will receive at least a $1/e$ fraction of the valuation that she deserves,
according to the PF solution.
This mechanism depends on a subroutine that computes the PF
allocation for the problem instance at hand; we therefore later study the running time of this subroutine, as
well as the robustness of our results in case this subroutine returns only approximate solutions.

The PA mechanism elicits the valuation function $v_i(\cdot)$ from each player $i$ and it computes the
PF allocation $x^*$ considering all the players' valuations. The final allocation 
$x$ 
output by the mechanism gives
each player $i$ only a fraction $f_i$ of 
her PF bundle $x_i^*$, i.e.\ for every item $j$ of which the PF allocation assigned to her a portion
of size $x^*_{ij}$, the PA mechanism instead assigns to her a portion of size
$f_i \cdot x^*_{ij}$, where
$f_i \in [0,1]$ depends on the extent to which the presence of player $i$ inconveniences the
other players; the value of $f_i$ may therefore vary across different players. The following
steps give a more precise description of the mechanism.
\bigskip		

\LinesNumbered	
\begin{algorithm}[H]
\SetEndCharOfAlgoLine{.}
Compute the PF allocation $x^*$ based on the reported bids\;
For each player $i$, compute the PF allocation $x^*_{-i}$ that would arise in her absence\;
Allocate to each player $i$ a fraction $f_{i}$ of everything that she receives according to $x^*$, where
 \begin{equation}\label{eq:fraction}
  f_i = \left(\frac{\prod_{i'\neq i}{[v_{i'}(x^*)]^{b_{i'}}}}{\prod_{i'\neq i}{[v_{i'}(x^*_{-i})]^{b_{i'}}}}\right)^{1/b_i}.
 \end{equation}
	\caption{The Partial Allocation mechanism.}
\end{algorithm}



\begin{lemma}\label{lem:feasible}
The allocation $x$ produced by the PA mechanism is feasible.
\end{lemma}
\begin{proof}
Since the PF allocation $x^*$ is feasible, to verify that the allocation produced by the PA mechanism
is also feasible, it suffices to show 
that $f_i \in [0,1]$ for every bidder $i$. The fact that $f_i\geq 0$ is clear since
both the numerator and the denominator are non-negative. To show that $f_i\leq 1$, note that
\[ x^*_{-i} = \arg\max_{x'\in \mathcal{F}}\left\{ \prod_{i'\neq i}{v_{i'}(x')} \right\}. \]
Since $x^*$ remains a feasible allocation ($x^*\in\mathcal{F}$) after removing bidder $i$ (we can just discard bidder
$i$'s share), this implies
\begin{equation*}
\prod_{i'\neq i}{v_{i'}(x^*)} \leq \prod_{i'\neq i}{v_{i'}(x^*_{-i})}.
\end{equation*}
\end{proof}

\subsection{Truthfulness}
We now show that, despite the fact that this mechanism is not dictatorial and does not use
monetary payments, it is still in the best interest of every player to report her true
valuation function, irrespective of what the other players do.

\begin{theorem}\label{thm:truth}
The PA mechanism is truthful.
\end{theorem}
\begin{proof}
In order to prove this theorem, we approach the PA mechanism from the perspective of some arbitrary player $i$.
Let $\bar{v}_{i'}(\cdot)$ denote the valuation function that each player $i' \neq i$ reports to the PA mechanism.
We assume that the valuation functions reported by these players may 
differ from their true ones, $v_{i'}(\cdot)$.
Player $i$ is faced with the options of, either reporting her true valuation function $v_i(\cdot)$,
or reporting some false valuation function $\bar{v}_i(\cdot)$.
After every player has reported some valuation function, the PA mechanism computes the PF allocation with respect to
these valuation functions; let $x_{{\text{\tiny T}}}$ denote the PF allocation that arises if player $i$ reports the truth
and $x_{{\text{\tiny L}}}$ otherwise. Finally, player $i$ receives a fraction of what the computed PF allocation assigned
to her, and how big or small this fraction will be depends on the computed PF allocation. Let $f_{{\text{\tiny T}}}$ denote the fraction
of her allocation that player $i$ will receive if $x_{{\text{\tiny T}}}$ is the computed PF allocation and $f_{{\text{\tiny L}}}$ otherwise.
Since the players have homogeneous valuation functions of degree one, what we need to show is
that $f_{{\text{\tiny T}}}~ v_i(x_{{\text{\tiny T}}}) \geq f_{{\text{\tiny L}}} ~v_i(x_{{\text{\tiny L}}})$, or equivalently that
\[  [f_{{\text{\tiny T}}}~ v_i(x_{{\text{\tiny T}}})]^{b_i} ~\geq~ [f_{{\text{\tiny L}}} ~v_i(x_{{\text{\tiny L}}})]^{b_i}.\]
Note that the 
denominators
of both
fractions $f_{{\text{\tiny T}}}$ and $f_{{\text{\tiny L}}}$, as 
given
by Equation~\eqref{eq:fraction}, will be the same since 
they are
independent of
the valuation function reported by player $i$. Our problem therefore reduces to proving that
\begin{equation}\label{ineq:truth}
[v_i(x_{{\text{\tiny T}}})]^{b_i} \cdot \prod_{i'\neq i}{[\bar{v}_{i'}(x_{{\text{\tiny T}}})]^{b_{i'}}} ~ \geq ~ [v_i(x_{{\text{\tiny L}}})]^{b_i} \cdot \prod_{i'\neq i}{[\bar{v}_{i'}(x_{{\text{\tiny L}}})]^{b_{i'}}}.
\end{equation}
To verify that this inequality holds we use the fact that the PF allocation is the one that maximizes the product
of the corresponding reported valuations.
This means that
\begin{equation*}
x_{{\text{\tiny T}}} = \arg\max_{x\in \mathcal{F}}\left\{ [v_i(x)]^{b_i} \cdot \prod_{i'\neq i}{[\bar{v}_{i'}(x)]^{b_{i'}}} \right\},
\end{equation*}
and since $x_{{\text{\tiny L}}} \in \mathcal{F}$, this implies that Inequality~\eqref{ineq:truth} holds, and hence reporting
her true valuation function is a dominant strategy for every player $i$.
\end{proof}

The arguments used in the proof of Theorem~\ref{thm:truth} imply that, given the valuation functions
reported by all the other players $i' \neq i$, player $i$ can effectively choose any bundle that 
she wishes, but for each bundle the mechanism defines what fraction 
player $i$ can keep. One can therefore think of the fraction of the bundle
thrown away as a form of non-monetary ``payment'' that the player must suffer in exchange for that bundle, 
with different bundles matched to different payments. The fact that the PA mechanism is truthful implies that these payments,
in the form of fractions, make the bundle allocated to her by allocation $x^*$ the most desirable one. 
We revisit this interpretation in Section~\ref{sec:withmoney}.

\subsection{Approximation}
Before studying the approximation factor of the PA mechanism, we first state a lemma 
which will be useful for proving Theorem~\ref{thm:approx} (its proof is deferred to the Appendix).

\begin{lemma}\label{lem:bound}
For any set of pairs $(\delta_i,\beta_i)$ with $\beta_i\geq 1$ and $\sum_i \beta_i \cdot \delta_i\leq b$
the following holds (where $B =\sum_i \beta_i$)
\[ \prod_i{(1+\delta_i)^{\beta_i}} \leq \left(1 + \frac{b}{B}\right)^B. \]
\end{lemma}

Using this lemma we can now prove tight bounds for the approximation factor of the Partial Allocation
mechanism. As we show in this proof, the approximation factor 
depends directly 
on the relative weights of the players.
For simplicity in expressing the approximation factor, let $b_{min}$
denote
the smallest value of $b_i$ across all bidders of an instance and 
let
$\bar{B}=\left(\sum_{i\in N}b_i \right)-b_{min}$
be the sum of the $b_i$ values of all the other bidders. Finally, let $\psi = \bar{B}/b_{min}$ denote
the ratio of these two values.

\begin{theorem}\label{thm:approx}
The approximation factor of the Partial Allocation mechanism for the class of problem 
instances of some given $\psi$ value is exactly
\begin{equation*}
\left(1 + \frac{1}{\psi}\right)^{-\psi}.
\end{equation*}
\end{theorem}
\begin{proof}
The PA mechanism allocates to each player $i$ a fraction $f_i$ of her PF allocation, and for
the class of homogeneous valuation functions of degree one this means that the final valuation of player
$i$ will be $v_i(x) = f_i \cdot v_i(x^*)$.
The approximation factor guaranteed by the mechanism is therefore equal to $\min_i\{f_i\}$.
Without loss of generality, let player $i$ be the one with the minimum value of $f_i$.
In the PF allocation $x^*_{-i}$ that the PA mechanism computes after removing player
$i$, every other player $i'$ experiences a value of $v_{i'}(x^*_{-i})$. Let $d_{i'}$
denote the proportional change between the valuation of player $i'$ for allocation $x^*$
and allocation $x^*_{-i}$, i.e.
\[ v_{i'}(x^*_{-i}) = (1+d_{i'}) v_{i'}(x^*).\]
Substituting for $v_{i'}(x^*_{-i})$ in Equation~\eqref{eq:fraction} yields:
\begin{equation}\label{eq:frac_prop}
f_i = \left(\frac{1}{\prod_{i'\neq i} (1+d_{i'})^{b_{i'}}}\right)^{1/b_i}.
\end{equation}
Since $x^*$ is a PF allocation, Inequality~\eqref{ineq:PF} implies that
\begin{align}
\sum_{i'\in N}{\frac{b_{i'}[v_{i'}(x^*_{-i})-v_{i'}(x^*)]}{v_{i'}(x^*)}} ~&\leq~ 0 ~~ \Longleftrightarrow \notag \\
\sum_{i'\neq i} b_{i'}d_{i'} ~+~  \frac{b_i[v_i(x^*_{-i})-v_i(x^*)]}{v_i(x^*)} ~&\leq~ 0 ~~ \Longleftrightarrow \notag \\
\sum_{i'\neq i} b_{i'}d_{i'} ~&\leq~ b_i. \label{eq:PF_implication}
\end{align}
The last equivalence holds due to the fact that $v_i(x^*_{-i})=0$, since allocation $x^*_{-i}$ clearly assigns
nothing to player $i$.

Let $B_{-i}=\sum_{i'\neq i}b_{i'}$;
using Inequality~\eqref{eq:PF_implication}
and Lemma~\ref{lem:bound}
(on substituting $b_i$ for $b$, $d_{i'}$ for $\delta_i$, $b_{i'}$ for $\beta_i$, and $B_{-i}$ for $B$), 
it follows from Equation~\eqref{eq:frac_prop} that
\begin{equation}\label{ineq:lb}
f_i \geq \left(1 + \frac{b_i}{B_{-i}}\right)^{-\frac{B_{-i}}{b_i}}.
\end{equation}

To verify that this bound is tight, consider any instance with just one item and 
the given $\psi$ value. The PF solution dictates that each player should
be receiving a fraction of the item proportional to the player's $b_i$ value.
The removal of a player $i$ therefore leads to a proportional increase of exactly
$b_i/B_{-i}$ for each of the other players' PF valuation. The PA mechanism therefore assigns to every
player $i$ a fraction of her PF allocation which is equal to the right hand side of
Inequality~\eqref{ineq:lb}. The player with the smallest $b_i$ value receives the
smallest fraction.
\end{proof}

The approximation factor of Theorem~\ref{thm:approx} implies that $f_i\geq 1/2$
for instances with two players having equal $b_i$ values, and $f_i \geq 1/e$
even when $\psi$ goes to infinity; we therefore get the following corollary.

\begin{corollary}\label{cor:approx}
The Partial Allocation mechanism always yields an allocation $x$
such that for every participating player $i$
\[ v_i(x) ~ \geq ~ \frac{1}{e} \cdot v_i(x^*). \]
\end{corollary}

To complement this approximation factor, we now provide a negative result showing
that, even for the special case of additive linear valuations, no truthful mechanism
can guarantee an approximation factor better than $\frac{n+1}{2n}$ for problem instances
with $n$ players.
\begin{theorem}\label{thm:LB1}
There is no truthful mechanism that can guarantee an approximation factor greater than
$\frac{n+1}{2n} + \eps$ for any constant $\eps >0$ for all $n$-player problem instances, even if the valuations
are restricted to being additive linear.
\end{theorem}
\begin{proof} 
For an arbitrary real value of $n > 1$, let $\rho =\frac{n+1}{2n}$, and assume that $Q$ is a truthful
resource allocation mechanism that guarantees a $(\rho+ \epsilon)$ approximation for all
$n$-player problem instances, where $\epsilon$ is a positive constant. This mechanism
receives as input the bidders' valuations and it returns a valid (fractional) allocation of the items.
We will define $n+1$ different input instances for this mechanism, each of which
will consist of $n$ bidders and $m=(k+1)n$ items, where $k>\frac{2}{\epsilon}$ will take very large values.
In order to prove the theorem, we will then show that $Q$ cannot simultaneously achieve
this approximation guarantee for all these instances, leading to a contradiction. For simplicity we
will refer to each bidder with a number from 1 to $n$, to each item with a
number from 1 to $(k+1)n$, and to each problem instance with a number from 1 to $n+1$.

We start by defining the first $n$ problem instances. For $i\leq n$, let 
problem instance
$i$ be as follows: Every bidder $i'\neq i$ has a valuation of $kn+1$ for item $i'$ and a
valuation of 1 for every other item; bidder $i$ has a valuation of 1 for all items.
In other words, all bidders except bidder $i$ have a strong preference for just one item,
which is different for each one of them.
The PF allocation for such additive linear valuations dictates that every bidder $i'\neq i$ is
allocated only item $i'$, while bidder $i$ is allocated all the remaining $kn+1$ items.
Since $Q$ achieves a $\rho+\epsilon$ approximation for this instance, it needs to provide
bidder $i$ with an allocation which the bidder values at least 
at
$\left(\rho+\epsilon\right)(kn+1)$.
In order to achieve this, mechanism $Q$ can assign to this bidder fractions of the set $M_{-i}$ of the
$n-1$ items that the PF solution allocates to the other bidders as well as fractions of the set $M_i$
of the $kn+1$ items that the PF allocation allocates to bidder $i$. Even if all of the $n-1$ items of
$M_{-i}$ were fully allocated to bidder $i$, the mechanism would still need to assign to this bidder
an allocation of value at least $\left(\rho+\epsilon\right)(kn+1)-(n-1)$ using items 
from
$M_i$. Since
$k>\frac{2}{\epsilon}$, 
$n-1<\frac{\epsilon}{2}(kn+1)$, and therefore mechanism $Q$ will need to allocate to bidder
$i$ a fractional assignment of items in $M_i$ that the bidder values at least
at
$\left(\rho +\frac{\epsilon}{2}\right)(kn+1)$. This implies that there must exist at least one item
in $M_i$ of which bidder $i$ is allocated a fraction of size at least $\left(\rho +\frac{\epsilon}{2}\right)$.
Since all the items in $M_i$ are identical and the numbering of the items is arbitrary, we can, without
loss of generality, assume that this item is item $i$. We have therefore shown that, for every instance
$i\leq n$ mechanism $Q$ will have to assign to bidder $i$ at least $\left(\rho +\frac{\epsilon}{2}\right)$
of item $i$, and an allocation of items in $M_i$ that guarantees her a valuation of at least
$\left(\rho +\frac{\epsilon}{2}\right)(kn+1)$.

We now define problem instance $n+1$, in which every bidder $i$ has a valuation of $kn+1$ for item $i$
and a valuation of 1 for all other items. The PF solution for this instance would allocate to each
bidder $i$ all of item $i$, as well as $k$ items from the set $\left\{n+1,...,(k+1)n\right\}$ (or more generally,
fractions of these items that add up to $k$). Clearly, every bidder $i$ can unilaterally misreport her valuation
leading to problem instance $i$ instead of this instance;
so, in order to maintain truthfulness, mechanism $Q$ will
have to provide every bidder $i$ of problem instance $n+1$ with at least the value that such a
deviation would provide her with. One can quickly verify that, even if mechanism $Q$ when faced with problem
instance $i$ provided bidder $i$ with no more than a $\left(\rho +\frac{\epsilon}{2}\right)$ fraction of item $i$,
still such a deviation would provide bidder $i$ with a valuation of at least
\[\left(\rho +\frac{\epsilon}{2}\right)(kn+1) ~+~ \left(\rho +\frac{\epsilon}{2}\right)kn ~\geq~\left(\rho +\frac{\epsilon}{2}\right)2kn.\]
The first term of the left hand side comes from the fraction of item $i$ that the bidder receives
and the second term comes from the average fraction of the remaining items. If we substitute 
$\rho =\frac{n+1}{2n}$,
we get that the truthfulness of $Q$ implies that every bidder $i$ of problem instance $n+1$ will have to receive an allocation of value at least
\[\left(\frac{n+1}{2n} +\frac{\epsilon}{2}\right)2kn ~=~ kn+k+\epsilon kn.\]
For any given constant value of $\epsilon$ though, since $k>\frac{2}{\epsilon}$ and $n>1$, every bidder will need to be assigned
an allocation that she values 
at
more than $kn+k+2$, which is greater than the valuation of $kn+k+1$ that
the player receives in the PF solution. This is obviously a contradiction since
the PF solution is Pareto efficient and there cannot exist any other allocation 
for
which all
bidders receive a strictly greater valuation.
\end{proof}

Theorem~\ref{thm:LB1} implies that, even if all the players have equal $b_i$ values, no truthful mechanism can
guarantee 
a
greater than $3/4$ approximation even for instances with just two bidders, and this bound drops further as the
number of bidders increases, finally converging to $1/2$. To complement the statement of Corollary~\ref{cor:approx},
we therefore get the following corollary.
\begin{corollary}\label{cor:inapprox}
No truthful mechanism can guarantee that it will always yield an allocation $x$
such that for any $\eps > 0$ and for every participating player $i$
\[ v_i(x) ~ \geq ~ \left( \frac{1}{2} + \eps \right) \cdot v_i(x^*). \]
\end{corollary}
\subsection{Envy-Freeness}
We now consider the question of whether the outcomes that the Partial Allocation mechanism yields are envy-free; we show
that, for two well studied types of valuation functions this is indeed the case, thus providing further
evidence of the fairness properties of this mechanism. We start by showing that, if the bidders 
have additive linear valuations, then the outcome that the PA mechanism outputs is also envy-free.
\begin{theorem}\label{thm:addef}
The PA mechanism is envy-free for additive linear valuations.
\end{theorem}
\begin{proof}
Let $x^*$ denote the PF allocation including all the bidders, with each bidder's
valuations scaled so that $v_i(x^*)=1$. Let $v_i(x^*_j)$ denote the value of bidder
$i$ for $x^*_j$, the PF share of bidder $j$ in $x^*$, and let $x^*_{-i}$ denote the PF allocation that
arises after removing some bidder $i$. The PA mechanism allocates each (unweighted) bidder
$i$ a fraction $f_i$ of her PF share, where
\begin{equation*}
f_i ~=~ \frac{\prod_{k\neq i}{[v_k(x^*)]}}{\prod_{k\neq i}{[v_k(x^*_{-i})]}} ~=~ \frac{1}{\prod_{k\neq i}{[v_k(x^*_{-i})]}}.
\end{equation*}

In order to prove that the PA mechanism is envy-free, we need to show that for every bidder $i$,
and for all $j\neq i$, $f_iv_i(x^*) \geq f_jv_i(x^*_j)$, or equivalently
\begin{equation}\label{ineq:efgoal}
\frac{1}{\prod_{k\neq i}{[v_k(x^*_{-i})]}} ~\geq~ \frac{v_i(x^*_j)}{\prod_{k\neq j}{[v_k(x^*_{-j})]}} ~\Leftrightarrow ~
\prod_{k\neq j}{[v_k(x^*_{-j})]} ~\geq~ v_i(x^*_j)\prod_{k\neq i}{[v_k(x^*_{-i})]}.
\end{equation}
To prove the above inequality, we will modify allocation $x^*_{-i}$ so as to create an allocation $x_{-j}$ such that
\begin{equation}\label{ineq:ef}
\prod_{k\neq j}{[v_k(x_{-j})]} ~\geq~ v_i(x^*_j)\prod_{k\neq i}{[v_k(x^*_{-i})]}.
\end{equation}
Clearly, for any feasible allocation $x_{-j}$ it must be the case that
\begin{equation}\label{ineq:ef2}
\prod_{k\neq j}{[v_k(x^*_{-j})]} ~\geq~ \prod_{k\neq j}{[v_k(x_{-j})]},
\end{equation}
since $x^*_{-j}$ is, by definition, the feasible allocation that maximizes this product.
Therefore, combining Inequalities~\eqref{ineq:ef} and~\eqref{ineq:ef2} implies Inequality~\eqref{ineq:efgoal}.
It remains to construct an allocation $x_{-j}$ satisfying Inequality~\eqref{ineq:ef}.

To construct allocation $x_{-j}$, we use allocation $x^*_{-i}$ and we define the following weighted directed 
graph $G$ based on $x^*_{-i}$: the set of vertices corresponds to the set of bidders, and a directed edge from 
the vertex for bidder $j$ to the vertex for bidder $k$ exists if and only if $x^*_{-i}$ allocates to bidder $j$ portions of items that were 
instead allocated to bidder $k$ in $x^*$. The weight of such an edge is equal to the total value that bidder $j$ 
sees in all these portions.
Since the valuations of all bidders are scaled so that $v_j(x^*)=1$ for all $j$, this implies that, if the weight
of some edge $(j,k)$ is $v$ (w.r.t.\ these scaled valuations), then the total value of bidder $k$ for those same
portions that bidder $j$ values at $v$, is at least $v$. If that were not the case, then $x^*$ would not have allocated
those portions to bidder $k$; allocating them to bidder $j$ instead would lead to a positive aggregate proportional 
change to the valuations. This means that we can assume, without loss of generality, that the graph is a directed 
acyclic one; if not, we can rearrange the allocation so as to remove any directed cycles from this graph without 
decreasing any bidder's valuation. 

Also note that for every bidder $k\neq i$ it must be the case that $v_k(x^*_{-i})\geq v_k(x^*)$.
To verify this fact, assume that it is not true, and let $k$ be the bidder with the minimum value
$v_k(x^*_{-i})$. Since $v_k(x^*_{-i})<v_k(x^*)=1$, it must be the case that $x^*_{-i}$ does not allocate
to bidder $k$ all of her PF share according to $x^*$, thus the vertex for bidder $k$ has incoming edges 
of positive weight in the directed acyclic graph $G$, and it therefore belongs to some directed path.
The very first vertex of this path is a source of $G$ that corresponds to some bidder $s$; the fact that this vertex has
no incoming edges implies that $v_s(x^*_{-i})\geq v_s(x^*) = 1$. Since $v_k(x^*_{-i})<1$ we can deduce
that there exists some directed edge $(\alpha, \beta)$ along the path from $s$ to $k$ such that
$v_\alpha(x^*_{-i})>v_\beta(x^*_{-i})$. Returning some of the portions contributing to this edge from 
bidder $\alpha$ to bidder $\beta$ will lead to a positive aggregate proportional change to the valuations, 
contradicting that $x^*_{-i}$ is the PF allocation excluding bidder $i$.
Having shown that $v_k(x^*_{-i})\geq v_k(x^*)$ for every bidder $k$ other than $i$, we can now deduce that
the total weight of incoming edges for the vertex in $G$ corresponding to any bidder $k\neq i$ is no more than 
the total weight of the outgoing edges. Finally, this also implies that the only sink of $G$ will have to be the 
vertex for bidder $i$.

The first step of our construction starts from allocation $x^*_{-i}$ and it reallocates some of the $x^*_{-i}$ allocation,
leading to a new allocation $\bar{x}$. Using the directed subtree of $G$ rooted at the vertex of bidder $j$, we
reduce to zero the weights of the edges leaving $j$ by reducing the allocation at $j$, increasing the allocation
at $i$, and suitably changing the allocation of other bidders. More
specifically, we start by returning all the portions that bidder $j$ was allocated in $x^*_{-i}$ but not
in $x^*$, back to the bidders who were allocated these portions in $x^*$. These bidders to whom some portions
were returned then return portions of equal value that they too were allocated in $x^*_{-i}$ but not in $x^*$;
this is possible since, for each such bidder, the total incoming edge weight of its vertex is no more than the 
total outgoing edge weight. We repeat this process until the sink, the vertex for bidder $i$, is reached. 
One can quickly verify that 
\begin{equation}\label{ineq:xbar}
v_i(\bar{x})\geq v_j(x^*_{-i}) - v_j(\bar{x});
\end{equation}
in words, the value that bidder $i$ gained
in this transition from $x^*_{-i}$ to $\bar{x}$ is at least as large as the value that bidder $j$ lost in that
same transition. Finally, in allocation $\bar{x}$, whatever value $v_j(\bar{x})$ bidder $j$ is left with comes only
from portions that were part of her PF share in $x^*$. 

Bidder $j$'s total valuation for any portions of her PF share in $x^*$ that are allocated to other bidders in 
$x^*_{-i}$ is equal to $1 - v_j(\bar{x})$. Thus, bidder $i$'s valuation for those same portions will be at most 
$1 - v_j(\bar{x})$; otherwise modifying $x^*$ by allocating these portions to $i$ would lead to a positive aggregate 
change to the valuations. This means that for bidder $i$ the portions remaining with bidder $j$ in allocation $\bar{x}$ 
have value at least $v_i(x^*_j)-(1 - v_j(\bar{x}))$. 
We conclude the construction of allocation $x_{-j}$ by allocating all the remaining portions allocated to bidder $j$ in $\bar{x}$
to bidder $i$, leading to
\begin{align*}
v_i(x_{-j}) &\geq v_i(\bar{x}) + v_i(x^*_j)-(1 - v_j(\bar{x}))\\
&\geq v_j(x^*_{-i}) - v_j(\bar{x}) + v_i(x^*_j)-(1 - v_j(\bar{x}))\\
&\geq v_j(x^*_{-i})-1 + v_i(x^*_j)\\
&\geq [v_j(x^*_{-i})-1] v_i(x^*_j)+ v_i(x^*_j)\\
&= v_j(x^*_{-i})v_i(x^*_j).
\end{align*}
The second inequality is deduced by substituting from Inequality~\eqref{ineq:xbar}; the last inequality can be verified by
using the fact that $v_i(x^*_j)\leq 1$, and multiplying both sides of this inequality with the non-negative value 
$v_j(x^*_{-i}) - 1$, leading to $[v_j(x^*_{-i})-1] v_i(x^*_j) \leq v_j(x^*_{-i})-1$.
Also note that for all $k\notin \{i,j\}$, $v_k(x_{-j})=v_k(x^*_{-i})$. We therefore conclude that
Inequality~\eqref{ineq:ef} is true. 
%
\end{proof}
Following the same proof structure we can now also show that the PA mechanism is envy-free when the bidders
have Leontief valuations.
\begin{theorem}
The PA mechanism is envy-free for Leontief valuations.
\end{theorem}
\begin{proof}
Just as in the proof of Theorem~\ref{thm:addef}, let $x^*$ denote the PF allocation including all the 
bidders, with each bidder's valuations scaled so that $v_i(x^*)=1$. Also, let $v_i(x^*_j)$ denote the value of bidder
$i$ for $x^*_j$, the PF share of bidder $j$ in $x^*$, and let $x^*_{-i}$ denote the PF allocation that
arises after removing some bidder $i$.

Following the steps of the proof of Theorem~\ref{thm:addef} we can reduce the problem of showing that the PA mechanism 
is envy-free to constructing an allocation $x_{-j}$ that satisfies
Inequality~\eqref{ineq:ef}, i.e.\ such that
\begin{equation*}
\prod_{k\neq j}{[v_k(x^*_{-j})]} ~\geq~ \prod_{k\neq j}{[v_k(x_{-j})]} ~\geq~ v_i(x^*_j)\prod_{k\neq i}{[v_k(x^*_{-i})]}.
\end{equation*}
To construct allocation $x_{-j}$, we start from allocation $x^*_{-i}$ and we reallocate the bundle of item fractions
allocated to bidder $j$ in $x^*_{-i}$ to bidder $i$ instead, while maintaining the same allocations for all other bidders.
Therefore, after simplifying the latter inequality using the fact that $v_k(x_{-j})=v_k(x^*_{-i})$ for all $k\neq i,j$, 
what we need to show is that 
\begin{equation}\label{ineq:leoef}
v_i(x_{-j}) ~\geq~ v_i(x^*_j) v_j(x^*_{-i}).
\end{equation}
Note that, given the structure of Leontief valuations, every bidder is interested in bundles of item fractions that satisfy
specific proportions. We can, without loss of generality, assume that the PF allocation allocates a fraction of some 
resource to a bidder only when this fraction leads to an increase of the bidder's valuation.
This means that the bundle of item fractions allocated to bidder $j$ in $x^*$ and the one allocated to
her in $x^*_{-i}$ both satisfy the same proportions; that is, there exists some constant $c$ such that, for each one of the
items, bidder $j$ receives in $x^*_{-i}$ exactly $c$ times the amount of that item that she receives in $x^*$. As a result, 
given the fact that Leontief valuations are homogeneous of degree one, $v_j(x^*_{-i})=c\cdot v_j(x^*) = c$ (using the fact that $v_j(x^*)=1$).
Similarly, since $x_{-j}$ allocates to bidder $i$ the bundle of bidder $j$ in $x^*_{-i}$, and using the homogeneous structure
of Leontief valuations, this implies that $v_i(x_{-j})=c\cdot v_i(x^*_j)$. Substituting these two equalities in 
Inequality~\eqref{ineq:leoef} verifies
that the inequality is true, thus concluding the proof.
\end{proof}

\subsection{Running Time and Robustness}

The PA mechanism has reduced the problem of truthfully implementing a constant factor approximation
of the PF allocation to computing 
exact PF allocations
for 
several
different problem instances,
as this is the only subroutine that the mechanism calls. If the valuation functions
of the players are affine, then there is a polynomial time algorithm to compute the
exact PF allocation~\cite{DevanurPSV08,JV07}.

We now show that, even if the PF solution can be only approximately computed in polynomial time,
our truthfulness and approximation related statements are robust with respect to such approximations
(all the proofs of this subsection are deferred to the Appendix).
More specifically, we assume that the PA mechanism uses a polynomial time algorithm that computes a
feasible allocation $\widetilde{x}$ instead of $x^*$ such that
\[
\left[ \prod_i [v_i(\widetilde{x})]^{b_i} \right]^{1/B} ~\geq~ 
  \left[ (1-\epsilon)  \prod_i [v_i(x^*)]^{b_i} \right]^{1/B}, ~~~  \text{where}~ B = \sum_{i=1}^n b_i.
  \]
Using this algorithm, the PA mechanism can be adapted as follows:
\LinesNumbered	
\begin{algorithm}
\SetEndCharOfAlgoLine{.}
Compute the approximate PF allocation $\widetilde{x}$ based on the reported bids\;
For each player $i$, compute the approximate PF allocation $\widetilde{x}_{-i}$ that would arise in her absence\;
Allocate to each player $i$ a fraction $\widetilde{f}_{i}$ of everything that she receives according to $\widetilde{x}$, where
 \begin{equation}\label{eq:approxfraction}
  \widetilde{f}_i = \min\left\{1 ~,~ \left(\frac{\prod_{i'\neq i}{[v_{i'}(\widetilde{x})]^{b_{i'}}}}{\prod_{i'\neq i}{[v_{i'}(\widetilde{x}_{-i})]^{b_{i'}}}}\right)^{1/b_i}\right\}.
 \end{equation}
	\caption{The Approximate Partial Allocation mechanism.}
\end{algorithm}

For this adapted version of the PA mechanism to remain feasible, we need to make sure that $\widetilde{f}_i$ remains less than or equal to 1.
Even if,
for some reason,
the allocation $\widetilde{x}_{-i}$ computed by the approximation algorithm does not 
satisfy this property,
the adapted mechanism 
will then choose
$\widetilde{f}_i=1$ instead.

We start by showing two lemmas verifying that this adapted version of the PA mechanism is robust both with respect to the approximation
factor it guarantees and with respect to the truthfulness guarantee.
\begin{lemma}\label{lem:adaptedapprox}
The approximation factor of the adapted PA mechanism for the class of problem 
instances of some given $\psi$ value is at least
\[ (1-\epsilon) \left(1 + \frac{1}{\psi}\right)^{-\psi}. \]
\end{lemma}
\begin{lemma}\label{adaptedtruth}
If a player misreports her preferences to the adapted PA mechanism, she may increase her valuation by at most a factor $(1-\epsilon)^{-2}$.
\end{lemma}




Finally, we show that if the valuation functions are, for example, concave
and homogeneous of degree one, then a feasible approximate PF allocation
can indeed be computed in polynomial time.

\begin{lemma}\label{lem:concave}
For concave homogeneous valuation functions of degree one, there exists an algorithm that computes a
feasible allocation $\widetilde{x}$ in time polynomial in $\log 1/\epsilon$ and the problem size, such that
\[\prod_i [v_i(\widetilde{x})]^{b_i} ~\geq~ (1-\epsilon)  \prod_i [v_i(x^*)]^{b_i}. \]
\end{lemma}

\subsection{Extension to General Homogeneous Valuations}
We can actually extend most of the results that we have shown for homogeneous valuation functions
of degree one to any valuation function that can be expressed as $v_i(f\cdot x)=g_i(f)\cdot v_i(x)$,
where $g_i(\cdot)$ is some increasing invertible function; for homogeneous valuation functions of degree $d$,
this function is $g_i(f)=f^d$. If this function is known for each bidder, we can then adapt the PA
mechanism as follows: instead of allocating to bidder $i$ a fraction $f_i$ of her allocation according to $x^*$ as defined in
Equation~\eqref{eq:fraction}, we instead allocate to this bidder a fraction $g_i^{-1}(f_i)$, where $g_i^{-1}(\cdot)$
is the inverse function of $g_i(\cdot)$.
If, for example, some bidder has a homogeneous valuation function of degree $d$, then allocating her
a fraction $f_i^{1/d}$ of her PF allocation has the desired effect and both truthfulness and the
same approximation factor guarantees still hold. The idea behind this transformation is that all that we
need in order to achieve truthfulness and the approximation factor is to be able to discard some fraction
of a bidder's allocation knowing exactly what fraction of 
her valuation this will correspond to.

\section{The Strong Demand Matching Mechanism}\label{sec:PFapprox}

The main result of the previous section shows that one can guarantee a good constant factor approximation
for any problem instance within a very large class of bidder valuations. The subsequent impossibility
result shows that, even if we restrict ourselves to problem instances with additive linear bidder valuations,
no truthful mechanism can guarantee more than a $1/2$ approximation.

In this section we study the question of whether one can achieve even better factors when restricted to
some well-motivated class of instances. We focus on additive linear valuations, and we provide a positive
answer to this question for problem instances where every item is highly demanded.
More formally, we consider problem instances for which the PF price (or equivalently
the competitive equilibrium price)
of every item is large when the budget of every player is fixed to one unit of scrip money\footnote{Remark: Our
mechanism does not make this assumption, but
the approximation guarantees are much better with this assumption.}.
The motivation behind this class of instances comes from problems
such as the one that arose with the Czech privatization auctions~\cite{AH2000}.
For such instances, where the number of players is much higher than the number of
items, one naturally anticipates that all item prices will be high in equilibrium.

For the rest of the chapter we assume that the weights of all players are equal and that their
valuations are additive linear.
Let $p^*_j$ denote the PF price of item $j$ when every bidder $i$'s budget $b_i$ is equal to 1.
Our main result in this section is the following:

\begin{theorem}\label{thm:SDM}
For additive linear valuations there exists a truthful mechanism that achieves
an approximation factor of $\, \min_j\left\{p^*_j/\lceil{p^*_j}\rceil\right\}$.
\end{theorem}
\noindent Note that if $k=\min_j p^*_j$, then this approximation factor is at least $k/(k+1)$.

We now describe our solution which we call the {\em Strong Demand Matching} mechanism (SDM).
Informally speaking, SDM starts by giving every bidder a unit amount of {\em scrip} money.
It then aims to discover {\em minimal} item prices such that the demand of each bidder at 
these prices can be satisfied using (a fraction of) just one item. In essense, our mechanism
is restricted to computing allocations that assign each bidder to just one item, and this 
restriction of the output space renders the mechanism truthful and gives an 
approximation guarantee much better than that of the PA mechanism for instances where every item is
highly demanded.


The procedure used by our mechanism is reminiscent of the method utilized by Demange et al.\
for multi-unit auctions~\cite{DemangeGS86}. Recall that this method increases the prices of all over-demanded items
uniformly until the set $R$ of over-demanded items changes, iterating this process until $R$ becomes
empty.
At that point, bidders are matched to preferred items.
For our setting, each bidder will seek to spend all her money, and we employ an analogous rising price methodology,
again making allocations when the set of over-demanded items is empty.
In our setting, the price increases are multiplicative rather than additive, however.
This approach also has some commonality with the
algorithm of Devanur et al.~\cite{DevanurPSV08} for computing the competitive equilibrium
for divisible items and bidders with additive linear valuations.
Their algorithm also proceeds by increasing the prices of over-demanded
items multiplicatively. Of course, their algorithm does not yield a truthful mechanism. Also, in order to 
achieve polynomial running time in computing the competitive equilibrium, their algorithm needs, at any one time, to be 
increasing the prices of a carefully selected subset of these items; this appears to make their 
algorithm quite dissimilar to ours. Next we specify our mechanism in more detail.

Let $p_j$ denote the price of item $j$, and let the \emph{bang per buck} that bidder $i$ gets
from item $j$ equal $v_{ij}/p_j$. We say that item $j$ is an MBB item of bidder $i$ if she gets the
maximum bang per buck from that item\footnote{Note that for each bidder there could be multiple MBB items
and that in the PF solution bidders are only allocated such MBB items.}.
For a given price vector $p$, let the demand graph $D(p)$ be a bipartite graph with bidders on one side and items 
on the other, such that there is an edge between bidder $i$ and item $j$ if and only if $j$ is an MBB item of bidder $i$.
We call $c_j=\lfloor p_j\rfloor$ the \emph{capacity} of item $j$ when
its price is $p_j$, and we say an assignment of bidders to items is \emph{valid}
if it matches each bidder to one of her MBB items and no item $j$ is matched to more than $c_j$ bidders.
Given a valid assignment $A$, we say an item $j$ is \emph{reachable} from bidder $i$ if there exists an alternating
path $(i, j_1, i_1, j_2, i_2,\cdots, j_k, i_k, j)$ in the graph $D(p)$ such that edges $(i_1, j_1),\cdots ,(i_k, j_k)$
lie in the assignment $A$. Finally, let $d(R)$ be the collection of bidders with all their MBB items in set $R$. 
Using these notions, we define the Strong Demand Matching mechanism in Figure~\ref{mech:SDM}.

\LinesNumbered	
\SetAlgoVlined
\begin{algorithm}[!ht]
\SetEndCharOfAlgoLine{.}
Initialize the price of every item $j$ to $p_j=1$\;
Find a valid assignment maximizing the number of matched bidders\;
\uIf{all the bidders are matched}{conclude with Step 15.}
Let $U$ be the set of bidders who are not matched in Step 2\;
Let $R$ be the set of all items reachable from bidders in the set $U$\;
Increase the price of each item $j$ in $R$ from $p_j$ to $r\cdot p_j$,
where $r\geq 1$ is the minimum value for which one of the following events takes place:\\
\uIf{the price of an item in $R$ reaches an integral value}{continue with Step 2.}
\uIf{the set of MBB items of some bidder $i \in d(R)$ increases, causing the set $R$ to grow}
{
	\uIf{for each item $j$ added to $R$, the number of bidders already matched to it equals $c_j$}{continue with Step 6.}
	\uIf{some item $j$ added to $R$ has $c_j$ greater than the number of bidders matched to it}{continue with Step 2.}
}
Bidders matched to some item $j$ are allocated a fraction $1/p_j$ of it.
	\caption{The Strong Demand Matching mechanism.}\label{mech:SDM}
\end{algorithm}

\subsection{Running time}
We first explain how to carry out Steps 6-14. Set $R$ can be computed using a breadth-first-search like algorithm.
To determine when the event of Step 8 takes place, we just need to know the smallest $\lceil p_j\rceil/p_j$
ratio over all items whose price is being increased.
For the event of Step 10, we need to calculate, for each bidder in $d(R)$, the ratio of the \emph{bang per buck}
for her MBB items and for the items outside the set $R$.

In terms of running time, if $c(R)=\sum_{j\in R} c_j$ denotes the total capacity in $R$, it is not difficult to see 
that if $U$ is non-empty, $|d(R)| > c(R)$. Note that each time either the event of Step 8 or the event of Step 13 occurs, 
$c(R)$ increases by at least 1, and thus, using the alternating path from a bidder in the set $U$ to the corresponding item, 
we can increase the number of matched bidders by at least 1; this means that this can occur at most $n$ times. 
The only other events are the unions (of connected components in graph $D(p)$) resulting from the event of Step 11.
Between successive iterations of either Step 8 or 13, there can be at most $\min(n,m)$ iterations of Step 11.
Thus there are $O(n\cdot \min(n,m))$ iterations of Step 11 overall and $O(n)$ iterations of Steps 8 and 13.

\subsection{Truthfulness and Approximation}

The proofs of the truthfulness and the approximation of the SDM mechanism use the following lemma
which states that the prices computed by the mechanism are the minimum prices supporting
a valid assignment. An analogous result was shown in~\cite{DemangeGS86}
for a multi-unit auction of non-divisible items.
We provide an algorithmic argument.

\begin{lemma}\label{lem:structure}
For any problem instance, if $p\geq 1$ is a set of prices for which there exists a valid assignment,
then the prices $q$ computed by the SDM mechanism will satisfy $q\leq p$.
\end{lemma}
\begin{proof}
Aiming for a contradiction, assume that $q_j > p_j$ for some item $j$, and let $\tilde{q}$ be the
maximal price vector that the SDM mechanism reaches before increasing the price of some item
$j'$ beyond $p_{j'}$ for the first time. In other words, $\tilde{q}\leq p$ and $\tilde{q}_{j'}=p_{j'}$.
Also, let $S = \{j\in M ~|~ \tilde{q}_j=p_j\}$, which implies that $\tilde{q}_j <p_j$
for all $j\notin S$. Clearly, any bidder $i$ who has MBB items in $S$ at prices $\tilde{q}$
will not be interested in any other item at prices $p$. This implies that the valid assignment
that exists for prices $p$ assigns every such bidder to one of her MBB items $j\in S$. Therefore,
the total capacity of items in $S$ at prices $\tilde{q}$ is large enough to support all these bidders
and hence no item in $S$ will be over-demanded at prices $\tilde{q}$. As a result, the SDM mechanism
will not increase the price of any item in $S$, which leads us to a contradiction.
\end{proof}

Using this lemma we can now prove the statements regarding the truthfulness and the approximation factor
of SDM; the following two lemmata imply Theorem~\ref{thm:SDM}.

\begin{lemma}
The SDM mechanism is truthful.
\end{lemma}

\begin{proof}
Given a problem instance, fix some bidder $i$ and let $x'$ and $q'$ denote the assignment and the prices
that the SDM mechanism outputs instead of $x$ and $q$ when this bidder reports a valuation vector $v'_i$
instead of her true valuation vector $v_i$.

If the item $j$ to which bidder $i$ is assigned in $x'$ is one of her MBB items w.r.t.\ her true valuations $v_i$
and prices $q'$, then $x'$ would be a valid assignment for prices $q'$ even if the bidder had not lied. 
Lemma~\ref{lem:structure} therefore implies that $q\leq q'$. Since the item to which bidder $i$ is assigned  
by $x$ is an MBB item and $q\leq q'$, we can conclude that $v_i(x)\geq v_i(x')$.

If on the other hand item $j$ is not an MBB item w.r.t.\ the true valuations of bidder $i$ and prices $q'$,
we consider an alternative valid assignment and prices. Starting from prices $q'$, we run the steps of the 
SDM mechanism assuming bidder $i$ has reported her true valuations $v_i$, and we consider the assignment 
$\bar{x}$ and the prices $\bar{q}$ that the mechanism would yield upon termination. Assignment $\bar{x}$ would 
clearly be valid for prices $\bar{q}$ if bidder $i$ had reported the truth; therefore Lemma~\ref{lem:structure} 
implies $q\leq \bar{q}$ and thus $v_i(x)\geq v_i(\bar{x})$. As a result, to conclude the proof it suffices to show 
that $v_i(\bar{x})\geq v_i(x')$. To verify this fact, we show that $q'_j=\bar{q}_j$, implying that $\bar{x}$ 
allocates to $i$ (a fraction of) some item which she values at least as much as a $1/q'_j$ fraction of item $j$. 

Consider the assignment $x'_{-i}$ that matches all bidders $i'\neq i$ according to $x'$ and leaves bidder $i$
unmatched. In the graph $D(q')$, if item $j$ is reachable from bidder $i$ given the valid assignment $x'_{-i}$, 
then all bidders would be matched by the very first execution of Step 1 of the mechanism. This is true because 
the capacity of item $j$ according to prices $q'$ is greater than the number of bidders matched to it in $x'_{-i}$. 
The alternating path $(i,j_1,i_1,j_2,i_2,\cdots,j_k,i_k,j)$ implied by the reachability can therefore be used to 
ensure that bidder $i$ is matched to an MBB item as well; this is achieved by matching $i$ to $j_1$, $i_1$ to $j_2$
and so on. Otherwise, if not all bidders can be matched in that very first step of the SDM mechanism, the mechanism can 
instead match the bidders according to $x'_{-i}$ and set $U=\{i\}$.\footnote{Note that this may not be the only way 
in which the SDM mechanism can proceed but, since the bidders' valuations for the final outcome are unique,
this is without loss of generality.} Before the price of item $j$ can be increased, Step 10 must add this item 
to the set $R$. If this happens though, item $j$ becomes reachable from bidder $i$ thus causing an alternating path
to form, and the next execution of Step 1 of the mechanism yields a valid assignment before $q'_j$ is ever increased.
\end{proof}

\begin{lemma}
The SDM mechanism achieves
an approximation factor of $\, \min_j\left\{p^*_j/\lceil{p^*_j}\rceil\right\}$.
\end{lemma}
\begin{proof}
%
We start by showing that there must exist a valid assignment at prices $fp^*$,
where $p^*$ corresponds to the PF prices and $f= \max_j \lceil p^*_j\rceil/p^*_j$.
Given any PF allocation $x^*$, we consider the bipartite graph on items and bidders that
has an edge between a bidder and an item if and only if $x^*$ assigns a portion of the item
to that bidder.
%
If there exists a cycle in this graph, one can remove an edge in this cycle by reallocating 
along the cycle while maintaining the valuation of every bidder. To verify that this is 
possible, note that all the items that a bidder is connected to by an edge are MBB items 
for this bidder, and therefore the bidder is indifferent regarding how her spending is 
distributed among them. Hence w.l.o.g.\ we can assume that the graph of $x^*$ is a forest.
 
For a given tree in this forest, root it at an arbitrary bidder. For each bidder in this 
tree, assign her to one of her child items, if any, and otherwise to her parent item.
Note that the MBB items for each bidder at prices $fp^*$ are the same as at prices $p^*$, 
so every bidder is assigned to one of her MBB items. Therefore, in order to conclude that 
this assignment is valid at prices $fp^*$ it is sufficient to show that the capacity constraints
are satisfied. The fact that $fp^*_j \geq \lceil{p^*_j}\rceil$ implies that 
$\left\lfloor fp^*_j \right\rfloor \geq \lceil{p^*_j}\rceil $, so we just need to show that,
for each item $j$, at most $\lceil{p^*_j}\rceil$ bidders are assigned to it.
To verify this fact, note that any bidder who is assigned to her parent item does not have child 
items so, in $x^*$, she is spending all of her unit of scrip money on that parent item. In other 
words, for any item $j$, the only bidder that may be assigned to it without having contributed to 
an increase of $j$'s PF price by 1 is the parent bidder of $j$ in the tree; thus, the total number
of bidders is at most $\lceil{p^*_j}\rceil$.

Now, let $q$ and $x$ denote the prices and the assignment computed by the SDM mechanism; 
by Lemma~\ref{lem:structure}, since there exists a valid assignment at prices $fp^*$, this 
implies that $q \leq fp^*$. The fact that the SDM mechanism assigns each bidder to one of 
her MBB items at prices $q$ implies that $v_i(x) = \max_j\{v_{ij}/q_j\}$. On the other hand, 
let $r$ be an MBB item of bidder $i$ at the PF prices $p^*$. If bidder $i$ had $b_i$ units 
of scrip money to spend on such MBB items, this would mean that $v_i(x^*)= b_i (v_{ir}/p^*_r)$ 
so, since $b_i=1$, this implies that $v_i(x^*)= v_{ir}/p^*_r$.  
Using this inequality along with the fact that $q_j \le fp^*_j$ for all items $j$, we can show that
\begin{equation*}
v_i(x) ~=~ \max_j\left\{\frac{v_{ij}}{q_j}\right\} ~\geq~ \frac{v_{ir}}{q_r} ~\geq~  \frac{v_{ir}}{fp^*_r} 
~=~ \frac{1}{f} \cdot v_i(x^*),
\end{equation*}
which implies that $v_i(x)\geq  \min_j\{p^*_j/\lceil p^*_j \rceil \} \cdot v_i(x^*)$ for any bidder $i$.
\end{proof}

\section{Connections to Mechanism Design with Money}\label{sec:withmoney}
In hindsight, a closer look at the mechanisms of this chapter reveals an interesting connection
between our work and known results from the literature on mechanism design {\em with money}. 
What we show in this section is that one can uncover useful interpretations of money-free mechanisms as 
mechanisms with actual monetary payments by instead considering appropriate logarithmic transformations 
of the bidders' valuations. In what follows, we expand on this connection for the two mechanisms that we 
have proposed.

\paragraph{Partial Allocation Mechanism}
We begin by showing that one can actually interpret the item fractions discarded by the Partial Allocation mechanism as VCG payments.
The valuation of player $i$ for the PA mechanism outcome is $v_i(x) = f_i \cdot v_i(x^*)$, or
\begin{equation}\label{eq:vcg}
v_i(x) ~=~ \left( \frac{\prod_{i'\neq i}{[v_{i'}(x^*)]^{b_{i'}}}}{\prod_{i'\neq i}{[v_{i'}(x^*_{-i})]^{b_{i'}}}}\right)^{1/b_i} \cdot v_i(x^*).
\end{equation}
Taking a logarithm on both sides of Equation~\eqref{eq:vcg} and then multiplying them by $b_i$ yields
\begin{equation}\label{eq:vcg2}
b_i \log v_i(x) = b_i \log v_i(x^*) -\left(\sum_{i'\neq i}{b_{i'}\log v_{i'}(x^*_{-i})} ~-~ \sum_{i'\neq i}{b_{i'}\log v_{i'}(x^*)} \right).
\end{equation}
Now, instead of focusing on each bidder $i$'s objective in terms of maximizing her valuation, we instead consider a logarithmic transformation of that
objective. More specifically, define $u_i(\cdot)=b_i \log v_i(\cdot)$ to be bidder $i$'s surrogate valuation. Since the logarithmic
transformation is an increasing function of $v_i$, for every bidder, her objective amounts to maximizing the value of this surrogate valuation.
Substituting in Equation~\eqref{eq:vcg2} using the surrogate valuation for each player gives 
\begin{equation*}
u_i(x) = u_i(x^*) - \left(\sum_{i'\neq i}{u_{i'}(x^*_{-i})} ~-~ \sum_{i'\neq i}{u_{i'}(x^*)} \right).
\end{equation*}
This shows that the surrogate valuation of a bidder for the output of the PA mechanism equals her
surrogate valuation for the PF allocation minus a ``payment'' which corresponds to exactly the externalities that 
the bidder causes with respect to the surrogate valuations! Note that, in settings where monetary payments 
are allowed, a VCG mechanism first computes an allocation that {\em maximizes the social welfare}, and then defines a set 
of monetary payments such that each bidder's payment corresponds to the externality that her presence causes. 
The connection between the PA mechanism and VCG mechanisms is complete if one notices that the PF 
objective aims to compute an allocation $x$ maximizing the value of $\sum_i{b_i \log {v_i(x)}}$, which is exactly
the social welfare $\sum_i{u_i(x)}$ with respect to the players' surrogate valuations. Therefore, the impact that the fraction being 
removed from each player's PF allocation has on that player's valuation is analogous to that of a VCG payment in the 
space of surrogate valuations. The fact that the PA mechanism is truthful can hence be deduced from the fact the
players wish to maximize their surrogate valuations and the VCG mechanism is truthful with respect to these valuations.
Nevertheless, the fact that the PA mechanism guarantees such a strong approximation of the PF solution remains surprising
even after revealing this reduction.

Also note that VCG mechanisms do not, in general, guarantee envy-freeness. The connection between the PA mechanism and VCG mechanisms
that we provide above, combined with the envy-freeness results that we proved for the PA mechanism for both additive linear and Leontief 
valuations, implies that the VCG mechanism is actually envy-free for settings with money and bidders having the corresponding surrogate 
valuations. Therefore, these results also contribute to the recent work on finding truthful, envy-free, and efficient mechanisms~\cite{CFFKO11, FL12}.

\paragraph{Strong Demand Matching Mechanism}
We now provide an even less obvious connection between the SDM mechanism and existing literature on mechanism
design with money; this time we illustrate how one can interpret the SDM mechanism as a stable matching mechanism. 
In order to facilitate this connection, we begin by reducing the problem of computing a valid assignment to the problem
of computing a ``stable'' matching:
we first scale each bidder's valuations so that her minimum non-zero valuation for an item is equal to $n$, and then, 
for each item $j$ we create $n$ copies of that item such that the $k$-th copy (where $k\in \{1,2,\dots,n\}$) of 
item $j$ has a reserve price $r_{jk}=k$. Given some price for each item copy, every buyer is seeking to be matched to 
one copy with a price that maximizes her valuation to price ratio, i.e.\ an MBB copy. A matching of each bidder to a
distinct item copy in this new problem instance is {\em stable} if and only if every bidder is matched to an MBB copy; it is
easy to verify that such a stable matching will always exist since there are $n$ copies of each item. Note that in a
stable matching any two copies of the same item, each of which is being matched to some bidder, need 
to have exactly the same price, otherwise the more expensive copy cannot be an MBB choice for the bidder matched to it. 

Now, a valid assignment of the initial input of the SDM mechanism implies a stable matching in the new problem instance:
set the price $p_{jk}$ of the $k$-th copy of item $j$ to be equal to the price $p_j$ of item $j$ in the valid assignment, 
unless this violates its reserve price, i.e.\ $p_{jk}=\max\{p_j , r_{jk}\}$, and match each bidder to a distinct copy
of the item that she was assigned to by the valid assignment; the validity of the assignment implies that, for each item $j$, 
the number of bidders assigned to it is at most $\lfloor p_j \rfloor$, and hence the number of item copies for which $p_{jk}\geq r_{jk}$,
i.e.\ $p_{jk}=p_j$ is enough to support all these bidders. 
Similarly, a stable matching of the item copies implies a valid assignment of the actual items of the initial problem instance: 
the price $p_j$ of each item $j$ is set to be equal to the minimum price over all its copies ($p_j=\min_k\{p_{jk}\}$), 
and each bidder who is matched to one of these copies is allocated a fraction $1/p_j$ of the corresponding actual item.


Using this reduction, we can now focus on the problem of computing such a stable matching of each bidder to just one distinct 
copy of some item; that is, we wish to define a price $p_{jk}\geq r_{jk}$ for each one of the $m\cdot n$ item copies, as 
well as a matching of each bidder to a distinct copy such that every bidder is matched to one of her MBB copies for the 
given prices. If we consider the same surrogate valuations $u_i(\cdot)=\log{v_i(\cdot)}$, the objective of each bidder $i$ to be matched to a copy of some item 
$j$ that maximizes the ratio $v_{ij}/p_{jk}$ is translated to the objective of maximizing the difference $\log{v_{ij}} - \log{p_{jk}}$.
If one therefore replaces the values $v_{ij}$ of the valuation vector reported by each bidder $i$ with the values $\log{v_{ij}}$,
then the initial problem is reduced to the problem of computing stable prices for these transformed valuations, assuming that monetary
payments are allowed. This problem has received a lot of attention in the matching literature, building upon the assignment model
of Shapley and Shubik~\cite{SS71}. Having revealed this connection, we know that we can truthfully compute a bidder optimal matching 
that does not violate the reserve prices using, for example, the mechanism of Aggarwal et al.~\cite{AMPP09}; one can verify that these
are exactly the logarithmic transformations of the prices of the SDM mechanism, and also that this is the matching the SDM mechanism 
computes. Note that increasing the surrogate prices of overdemanded item copies by some additive constant corresponds to increasing 
the corresponding actual prices by a multiplicative constant. Therefore, this transformation also sheds some light on why the SDM 
mechanism uses multiplicative increases of the item prices.

\section{Conclusion and Open Problems}

Our work was motivated by the fact that no incentive compatible mechanisms were 
known for the natural and widely used fairness concept of Proportional Fairness. 
In hindsight, our work provides several new contributions. First, the class of
bidder valuation functions for which our results apply is surprisingly large and
it contains several well studied functions; previous truthful mechanisms for 
fairness were studied for much more restricted classes of valuation functions. 
Second, to the best of our knowledge, this is first work that defines and 
gives guarantees for a strong notion of approximation for fairness, where one 
desires to approximate the valuation of every bidder. Finally, 
our Partial Allocation mechanism can be seen as a framework for designing truthful 
mechanisms without money. This mechanism can be generalized further by restricting 
the range of the outcomes (similar to maximal-in-range mechanisms when one can use 
money). Specifically, the set of feasible outcomes $\mathcal{F}$ can be restricted
to any downward closed subset of outcomes; that is, as long as 
$(x_{\!_1}, x_{\!_2}, ..., x_n)\in \mathcal{F}$ implies 
$(f_{\!_1} \cdot x_{\!_1},f_{\!_2}\cdot x_{\!_2}, ...,f_n\cdot x_n)\in \mathcal{F}$ 
for every set of scalars $f_i\in[0,1]$, then the mechanism remains well defined.
We believe that this generalization is a powerful one, and might allow for new solutions 
to other mechanism design problems without money.

In terms of open problems, the obvious one is to close the gap between the approximation
guarantee of Theorem~\ref{thm:approx} and the inapproximability result of Theorem~\ref{thm:LB1}.
According to these bounds, when all the bidders have equal $b_i$ values, the best possible 
approximation guarantee lies somewhere between $0.5$ and $0.75$ for two-bidder instances and
between $0.368$ and $0.5$ as the number of bidders goes to infinity.

Possibly the most interesting open problem though is the study of the
following natural objective: instead of aiming to maximize the minimum $v_i(x)/v_i(x^*)$
ratio across every bidder $i$, one may instead wish to maximize the product of all these ratios.
Note that maximizing this objective is equivalent to maximizing the PF objective $\prod_i v_i(x)$.
The Partial Allocation mechanism guarantees a $1/e$ approximation of the form
\begin{equation}\label{ineq:obj}
\left(\prod_i v_i(x)\right)^{1/n}\geq \frac{1}{e}\left(\prod_i v_i(x^*)\right)^{1/n}.
\end{equation}
On the other hand, the inapproximability result of Theorem~\ref{thm:LB1} does not apply to this objective
and hence one might hope to significantly improve the guarantee of Inequality~\eqref{ineq:obj}. The way to
do this would be to possibly sacrifice the value of some bidders, something that the objective studied in
this paper would not allow, in favor of this aggregate measure. Alternatively, one could prove stronger
inapproximability results showing that no such truthful mechanism exists.

Finally, a broader question that arises from this work has to do with the power of ``money burning''.
Specifically, one can verify that, in dealing with scale-free objectives such as the one studied in
this work and the one proposed above, discarding fractions of the bidders' resources allows the mechanism
designer to not worry about the actual scale of each bidder's valuations: the assumption of homogeneity
is sufficient for designing truthful mechanisms with non-trivial approximation guarantees. If, on the other
hand, discarding resources were disallowed and the mechanism designer were restricted to using only monetary 
payments, then homogeneity would not be sufficient and the scale of each bidder's valuations would need
to be elicited somehow before the appropriate payments could be chosen; this significantly complicates the
work of the mechanism designer. It would therefore be very interesting to better understand the settings for 
which ``money burning'' may lead to improved results despite the inefficiencies that it introduces.

\section{Acknowledgments} The second author would like to thank Vasilis Syrgkanis for his help in
clarifying the connection between the Partial Allocation mechanism and VCG payments.

\section*{APPENDIX}
\appendix
\renewcommand*{\thesection}{\Alph{section}}

\section{Omitted Proofs}\label{app:missing}

This Appendix includes the proofs that are missing from the main section.

\begin{proof}[of Lemma \ref{lem:bound}]

We first prove that this lemma is true for any number $k$ of pairs when $\beta_i=1$ for every pair.
For this special case we need to show that, if $\sum_{i=1}^k{\delta_i}\leq b$, then
\[ \prod_{i=1}^k{(1+\delta_i)} \leq \left(1 + \frac{b}{k}\right)^k .\]

Let $\bar{\delta}_i$ denote the values that actually maximize the left hand side of this inequality and
$\Delta_{k'}=\sum_{i=1}^{k'}\bar{\delta}_i$ denote the sum of these values up to $\bar{\delta}_{k'}$.
Note that it suffices to show that $\bar{\delta}_i = b/k$ for all $i$ since we have
\begin{equation*}
 \prod_{i=1}^k{(1+\delta_i)} \leq  \prod_{i=1}^k{(1+\bar{\delta}_i)},
\end{equation*}
and replacing 
$\bar{\delta}_i$
with
$b/k$ yields the inequality that we want to prove.

To prove that $\bar{\delta}_i=b/k$ we first prove that for any $k'\leq k$ and any
$i\leq k'$ we get $\bar{\delta}_i =\Delta_{k'}/k'$; we prove this fact by induction on $k'$:
For the basis step ($k'=2$) we show that $\bar{\delta}_1= \Delta_{2}/2$.
For any given value of $\Delta_2$ we know that any choice of $\delta_1$ will yield
\[ \prod_{i=1}^2{(1+\delta_i)} ~=~ (1+\delta_1)(1+\Delta_2 -\delta_1). \]

Taking the partial derivative with respect to $\delta_1$ readily shows that
this is maximized when $\delta_1 = \Delta_2/2$, thus $\bar{\delta}_1 = \Delta_2/2$.
For the inductive step we assume that $\bar{\delta}_i =\Delta_{k'-1}/(k'-1)$ for all $i\leq k'-1$.
This implies that for any given value of $\Delta_{k'}$, given a choice of $\delta_{k'}$
the remaining product is maximized if the following holds
\begin{equation*}
\prod_{i=1}^{k'}{(1+\delta_i)} =\left(1+\frac{\Delta_{k'}-\delta_{k'}}{k'-1}\right)^{k'-1}(1+\delta_{k'}).
\end{equation*}
Once again, taking the partial derivative of this last formula with respect to $\delta_{k'}$ for
any given $\Delta_{k'}$ shows that this is maximized when $\delta_{k'}=\Delta_{k'}/k'$. This of course
implies that $\Delta_{k'-1}=\frac{k'-1}{k'}\Delta_{k'}$ so $\bar{\delta}_i = \Delta_{k'}/k'$ for all $i\leq k'$.

This property of the $\bar{\delta}_i$ that we just proved, along with the fact that $\Delta_k\leq b$ implies
\begin{equation*}
 \prod_{i=1}^k{(1+\delta_i)} ~\leq~  \left(1 + \frac{\Delta_k}{k} \right)^k ~\leq~ \left(1 + \frac{b}{k} \right)^k.
\end{equation*}

We now use what we proved above in order to prove the lemma for
any rational $\delta_i$
using a proof by contradiction.
Assume that there exists a multiset $\mathcal{A}$ of pairs
$(\delta_i,\beta_i)$ with $\beta_i\geq 1$ and $\sum_i \beta_i \cdot \delta_i\leq b$
such that
\begin{equation}\label{ineq:contradiction}
\prod_i{(1+\delta_i)^{\beta_i}} > \left(1 + \frac{b}{B}\right)^B,
\end{equation}
where $B =\sum_i \beta_i$.
Let $M$ be an arbitrarily large value such that $\beta'_i=M\beta_i$ is a natural number for all $i$.
Also, let $b'=Mb$.
Then
$\sum_i \beta'_i \cdot \delta_i\leq b'$, and $B'=M\cdot B = \sum_i \beta'_i$.
Raising both sides of Inequality~\ref{ineq:contradiction}
to the power of $M$
yields
\[ \prod_i{(1+\delta_i)^{\beta'_i}} > \left(1 + \frac{b'}{B'}\right)^{B'}. \]
To verify that this is a contradiction, we create a multiset to which, for any pair $(\delta_i,\beta_i)$
of multiset $\mathcal{A}$, we add $\beta'_i$ pairs $(\delta_i,1)$. This multiset contradicts what we
showed above for the special case of pairs with $\beta_i=1$.

Extending the result to real valued $\delta_i$ just requires approximating the
$\delta_i$ closely enough with rational valued terms.
Specifically, let $\delta_i = \delta'_i + \epsilon_i$, where $\epsilon_i \ge 0$ and
$\delta'_i$ is rational.
Then $\sum_i \delta'_i \beta_i \le b$,
and by the result for rational $\delta$,
\[ \prod_i{(1+\delta'_i)^{\beta_i}} \le \left(1 + \frac{b}{B}\right)^B. \]
But then
\begin{eqnarray*}
\prod_i{(1+\delta_i)^{\beta_i}} & \le & \prod_i{(1+\delta'_i + \epsilon_i)^{\beta_i}} \\
& \le & \prod_i{ \left[ ( 1+\delta'_i) \left( 1 + \frac{\epsilon_i} {1+\delta'_i} \right) \right] ^{\beta_i}} \\
& \le & \left(1 + \frac{b}{B}\right)^B \prod_i{ \left( 1 + \frac{\epsilon_i} {1+\delta'_i} \right)  ^{\beta_i} }.
 \end{eqnarray*}
Since $\epsilon_i$ can be arbitrarily small, it follows that even for real valued $\delta_i$
\[ \prod_i{(1+\delta_i)^{\beta_i}} \le \left(1 + \frac{b}{B}\right)^{B}. \]
\end{proof}

\begin{proof}[of Lemma~\ref{lem:adaptedapprox}]
For any given approximate PF allocation $\widetilde{x}$, one can quickly verify that the valuation of bidder $i$ for her final allocation
only decreases as the value of $\prod_{i'\neq i}{[v_{i'}(\widetilde{x}_{-i})]^{b_{i'}}}$ increases. We can therefore assume that
the approximation factor is minimized
when the
denominator of Equation~\eqref{eq:approxfraction} takes on its maximum value, i.e. $\widetilde{x}_{-i}=x^*_{-i}$.
This implies that
the fraction in this equation will always be less than or equal to 1, and the valuation of bidder $i$ will therefore equal
\begin{align*}
\widetilde{f}_i \cdot v_i(\widetilde{x}) ~&\geq~ \left(\frac{\prod_{i'}{[v_{i'}(\widetilde{x})]^{b_{i'}}}}{\prod_{i'\neq i}{[v_{i'}(x^*_{-i})]^{b_{i'}}}}\right)^{1/b_i}\\
&\geq~(1-\epsilon)\left(\frac{\prod_{i'}{[v_{i'}(x^*)]^{b_{i'}}}}{\prod_{i'\neq i}{[v_{i'}(x^*_{-i})]^{b_{i'}}}}\right)^{1/b_i} \\
&=~ (1-\epsilon) f_i \cdot v_i(x^*).
\end{align*}
The first inequality holds because the right hand side is minimized when $\widetilde{x}_{-i}=x^*_{-i}$, and the second inequality holds because $\widetilde{x}$
is defined to be an allocation that approximates $x^*$.
%
The result follows on using Theorem~\ref{thm:approx}
to lower bound $f_i$.
%
\end{proof}

\begin{proof}[of Lemma~\ref{adaptedtruth}]
In the proof of the previous lemma we showed that, if bidder $i$ is truthful, then her valuation in the final allocation produced
by the adapted PA mechanism will always be at least $(1-\epsilon)$ times the valuation $f_i\cdot v_i(x^*)$ that she would receive if all
the PF allocations could be computed optimally rather than approximately. We now show that her valuation cannot be more
than $(1-\epsilon)^{-1}$ times greater than $f_i\cdot v_i(x^*)$, even if she misreports her preferences. Upon proving this statement,
the theorem follows from the fact that, even if
bidder $i$
being truthful
results in the worst possible approximation for this bidder,
still any lie can
increase her valuation by a factor of
at most
$(1-\epsilon)^{-2}$.

For any allocation $\widetilde{x}$ we know that $\prod_{i'}{[v_{i'}(\widetilde{x})]^{b_{i'}}}\leq \prod_{i'}{[v_{i'}(x^*)]^{b_{i'}}}$, by definition of PF.
Also, any allocation $\widetilde{x}_{-i}$ that the approximation algorithm may compute instead of $x^*_{-i}$ will satisfy
$\prod_{i'\neq i}{[v_{i'}(\widetilde{x})]^{b_{i'}}}\geq (1-\epsilon) \prod_{i'\neq i}{[v_{i'}(x^*)]^{b_{i'}}}$.
Using Equation~\eqref{eq:approxfraction} we can thus infer that
no matter what the
computed allocations $\widetilde{x}$ and $\widetilde{x}_{-i}$ are,
bidder $i$
will experience a valuation of at most
\begin{align*}
\left(\frac{\prod_{i'}{[v_{i'}(\widetilde{x})]^{b_{i'}}}}{\prod_{i'\neq i}{[v_{i'}(\widetilde{x}_{-i})]^{b_{i'}}}}\right)^{1/b_i} ~
&\leq~ \left(\frac{\prod_{i'}{[v_{i'}(x^*)]^{b_{i'}}}}{\prod_{i'\neq i}{[v_{i'}(\widetilde{x}_{-i})]^{b_{i'}}}}\right)^{1/b_i}\\
&\leq~(1-\epsilon)\left(\frac{\prod_{i'}{[v_{i'}(x^*)]^{b_{i'}}}}{\prod_{i'\neq i}{[v_{i'}(x^*_{-i})]^{b_{i'}}}}\right)^{1/b_i} \\
&\leq~ (1-\epsilon)~ f_i \cdot v_i(x^*).
\end{align*}
\end{proof}

\begin{proof}[of Lemma~\ref{lem:concave}]
As the valuation functions are all concave and homogeneous
of degree one, so is the following product,
\begin{equation*}
\left(  \prod_i [v_i(x)]^{b_i }\right)^{1/B}.
\end{equation*}
Also, note that this product has the same optima as the PF objective.
Consequently the above optimization is an instance of convex programming with linear
constraints, which can be solved approximately in polynomial time.
More precisely, an approximation with an additive error of $\epsilon$ to the optimal product of the
valuations can be found in time polynomial in the problem instance size and $\log(1/\epsilon)$~\cite{Nemirovski06}.
In addition, the approximation is a feasible allocation.

We normalize the individual valuations to have a value 1 for an allocation of everything.
If $B=\sum_i b_i$ is the sum of the bidders' weights then, at the optimum,  bidder $i$
has valuation at least $b_i/B$. To verify that this is true, just note that the sum of the prices of all
goods in the competitive equilibrium 
will be $B$ and
bidder 
$i$
will have a budget
of $b_i$. Since 
each
bidder will spend all her budget on the items she values the most for the prices
at hand, her valuation for her bundle will have to be at least $b_i/B$. This implies that
the optimum product valuation is at least $\prod_i (b_i/B)^{b_i/B} \ge \min_i b_i /B$;
this can be approximated to within an additive factor $\epsilon \cdot \min_i b_i/B$
in time polynomial in $ \log 1/\epsilon + \log B$,
and this is an approximation to within a multiplicative factor of $1 - \epsilon$.
\end{proof}

\bibliographystyle{elsarticle-num-names}
\bibliography{MARA}



\end{document}